\documentclass[11pt,a4paper]{article}
\usepackage{a4wide}
\usepackage{geometry}
\usepackage{showlabels}
\usepackage[utf8]{inputenc}
\usepackage{enumerate}
\usepackage{amsmath}
\usepackage{amsfonts}
\usepackage{nicefrac}
\usepackage{amssymb}
\usepackage{amsthm}
\usepackage{mathtools}
\usepackage{tikz}
\usetikzlibrary{calc}
\usepackage{pgfplots}
\pgfplotsset{compat=newest}
\usepackage{pgfplotstable}
\usepackage{xcolor}
\usepackage{authblk}
\usepackage{tikz}

\pgfplotsset{
    discard if not/.style 2 args={
        x filter/.code={
            \edef\tempa{\thisrow{#1}}
            \edef\tempb{#2}
            \ifx\tempa\tempb
            \else
                
            \fi
        }
    }
}
\pgfmathdeclarefunction{lg2}{1}{%
    \pgfmathparse{ln(#1)/ln(2)}%
}
\pgfmathdeclarefunction{lg10}{1}{%
    \pgfmathparse{ln(#1)/ln(10)}%
}

\usepackage{tabularx}
\usepackage{booktabs}
\usepackage{paralist}

\usepackage{algorithm2e}
\DontPrintSemicolon

\usepackage[pagebackref,pdfdisplaydoctitle,menucolor=orange!40!black,filecolor=magenta!40!black,urlcolor=blue!40!black,linkcolor=red!40!black,citecolor=green!40!black,colorlinks]{hyperref}

\renewcommand*{\backref}[1]{}
\renewcommand*{\backrefalt}[4]{%
	\ifcase #1%
	\or [p.~#2.]%
	\else [pp.~#2.]%
	\fi%
}
\usepackage[square,numbers]{natbib}
\usepackage[nameinlink,sort&compress,capitalize]{cleveref}
\newcommand{\ExternalLink}{%
    \tikz[x=1.2ex, y=1.2ex, baseline=-0.05ex]{%
        \begin{scope}[x=1ex, y=1ex]
            \clip (-0.1,-0.1) 
                --++ (-0, 1.2) 
                --++ (0.6, 0) 
                --++ (0, -0.6) 
                --++ (0.6, 0) 
                --++ (0, -1);
            \path[draw, 
                line width = 0.5, 
                rounded corners=0.5] 
                (0,0) rectangle (1,1);
        \end{scope}
        \path[draw, line width = 0.5] (0.5, 0.5) 
            -- (1, 1);
        \path[draw, line width = 0.5] (0.6, 1) 
            -- (1, 1) -- (1, 0.6);
        }
    }

\usepackage[backgroundcolor=gray!10,textsize=footnotesize]{todonotes}

\newtheorem{theorem}{Theorem}
\newtheorem{lemma}[theorem]{Lemma}

\newtheorem{corollary}[theorem]{Corollary}

\theoremstyle{definition}

\newtheorem{construction}[theorem]{Construction}

\crefname{rrule}{Rule}{Rules}

\newcommand{\prob}[1]{\textnormal{\textsc{#1}}}

\newcommand{\boxproblem}[4]{
    \begin{center}   
        \fbox{~\begin{minipage}{.97\textwidth}
            \vspace{2pt} 
            \noindent
            \normalsize\textsc{#1}
            \vspace{1pt}

            \setlength{\tabcolsep}{3pt}
            \renewcommand{\arraystretch}{1.0}
            \begin{tabularx}{\textwidth}{@{}lX@{}}
                \normalsize\textbf{Input:}       & \normalsize#2 \\
                \normalsize\textbf{Question:}    & \normalsize#3 \\
                \normalsize\textbf{Parameter: }  & \normalsize#4
            \end{tabularx}
        \end{minipage}}
    \end{center}
}

\DeclarePairedDelimiterX{\abs}[1]{\lvert}{\rvert}{#1}
\DeclarePairedDelimiterX{\norm}[1]{\lVert}{\rVert}{#1}
\DeclarePairedDelimiterX{\ceil}[1]{\lceil}{\rceil}{#1}

\newcommand{\NN}{\mathbb{N}}

\newcommand{\XP}{\ensuremath{\textnormal{XP}}}
\newcommand{\Wone}{\ensuremath{\textnormal{W[1]}}}
\newcommand{\FPT}{\ensuremath{\textnormal{FPT}}}
\newcommand{\NP}{\ensuremath{\textnormal{NP}}}
\newcommand{\coNP}{\ensuremath{\textnormal{coNP}}}
\newcommand{\poly}{\ensuremath{\textnormal{poly}}}
\newcommand{\NPincoNPslashpoly}{\ensuremath{\NP\subseteq \coNP/\poly}}

\newcommand{\bigO}{\mathcal{O}}

\newcommand{\oneto}[1]{[ #1 ]} 

\DeclareMathOperator{\vc}{vc}
\DeclareMathOperator{\fvs}{fvs}
\DeclareMathOperator{\fvn}{fvs}
\DeclareMathOperator{\tw}{tw}

\DeclareMathOperator{\cvd}{cd}

\newcommand{\calF}{\mathcal{F}}

\title{\Large\bf Vertex Cover and Feedback Vertex Set Above and Below Structural Guarantees\footnote{This work was initiated at the research retreat of the Algorithmics and Computational Complexity group, TU Berlin, in 2021.}}

\author{Leon Kellerhals}
\author{Tomohiro Koana\thanks{Supported by the DFG Project DiPa, NI 369/21.}}
\author{Pascal Kunz\thanks{Supported by the DFG Research Training Group 2434 ``Facets of Complexity''.}}

\affil{\small
  Technische Universit\"at Berlin, Faculty~IV, Institute of Software Engineering and Theoretical Computer Science, Algorithmics and Computational Complexity.\protect\\
  \texttt{\{leon.kellerhals,tomohiro.koana,p.kunz.1\}@tu-berlin.de}}

\date{}

\begin{document}

\maketitle

\begin{abstract}
	\noindent
	\textsc{Vertex Cover} parameterized by the solution size~$k$ is the quintessential fixed-parameter tractable problem.
	FPT algorithms are most interesting when the parameter is small.
	Several lower bounds on $k$ are well-known, such as the maximum size of a matching.
	This has led to a line of research on parameterizations of \textsc{Vertex Cover} by the difference of the solution size $k$ and a lower bound.
	The most prominent cases for such lower bounds for which the problem is FPT are the matching number or the optimal fractional LP solution.
	We investigate parameterizations by the difference between $k$ and other graph parameters including the feedback vertex number, the degeneracy, cluster deletion number, and treewidth with the goal of finding the border of fixed-parameter tractability for said difference parameterizations.
	We also consider similar parameterizations of the \textsc{Feedback Vertex Set} problem.
\end{abstract}

\section{Introduction}

Given an undirected graph $G$ and an integer $k$, the \textsc{Vertex Cover} problem asks whether there is a set of at most $k$ vertices that contains at least one endpoint of each edge.
\textsc{Vertex Cover} is arguably the most well-studied problem in parameterized complexity.
After significant efforts, the state-of-the-art FPT algorithm parameterized by the solution size $k$ runs in time~$\bigO(1.2738^k + kn)$ \cite{DBLP:journals/tcs/ChenKX10}, where $n$ is the number of vertices.
Very recently, Harris and Narayanaswamy~\cite{Harris2022} have announced an even faster algorithm with running time~$\bigO(1.2540^k \cdot n^{\bigO(1)})$

The aforementioned FPT algorithm is only useful when the parameter~$k$ is small.
In practice, however, the minimum vertex cover size is often large.
For this reason, many recent studies look into \textsc{Vertex Cover} where the parameterization is $k$ minus a lower bound on $k$.
For instance, if the maximum matching size $m$ is greater than $k$, then this would be a trivial no instance, since a vertex cover must contain at least one endpoint of each edge in any matching.
This naturally gives rise to the ``above guarantee''~\cite{DBLP:journals/jal/MahajanR99} parameter~$k - m$.
\textsc{Vertex Cover} is FPT with respect to $k - m$ \cite{DBLP:journals/jcss/RazgonO09}.
It has also been shown that \textsc{Vertex Cover} is FPT for even smaller above guarantee parameters such as $k - r$ \cite{DBLP:journals/toct/CyganPPW13,DBLP:journals/talg/LokshtanovNRRS14} and $k - 2 r + m$ \cite{Garg2016}, where $r$ is the optimal LP relaxation value of \textsc{Vertex Cover}.
Kernelization with respect to these parameters has also been studied \cite{Kratsch2018,DBLP:journals/jacm/KratschW20}.

This work considers above guarantee parameterizations of \textsc{Vertex Cover} where the lower bounds are structural parameters not related to the matching number, such as feedback vertex number, degeneracy, and cluster vertex deletion number.
We also study similar above guarantee parameterizations of the \textsc{Feedback Vertex Set} problem:
Given a graph $G$ and an integer~$k$, it asks whether there is a set of at most $k$ vertices whose deletion from $G$ results in a forest.
In this work, we do not deeply look into the ``below guarantee'' parameterization (where the number of vertices $n$ is the most obvious upper bound) because \textsc{Vertex Cover} and \textsc{Feedback Vertex Set} are known to be W[1]-hard when parameterized by $n - k$.\footnote{These two parameterized problem are essentially the \Wone-hard problems \textsc{Independent Set}~\cite{Downey1995} and \textsc{Maximum Induced Forest}~\cite{Khot2002}, respectively.}

\paragraph*{Motivation.}
We believe that FPT algorithms with above guarantee parameterizations may help explain the efficiency of some branching algorithms in practice.
Consider an instance $I = (G, k)$ of \textsc{Vertex Cover} for a complete graph $G$.
This instance is trivial to solve: $I$ has a solution if and only if $k \ge n - 1$.
This triviality, however, is overlooked by the worse-case running time bound of FPT algorithms parameterized by the solution size $k$ or the aforementioned smaller parameters $k - m$, $k - \ell$, or $k - 2 \ell + m$, all of which amount to $n / 2$ (for even $n$).
Now consider another above guarantee parameter~$k - (\omega - 1)$, where $\omega$ is the maximum clique size.
(Note that for a clique $C$, a vertex cover contains at least $|C| - 1$ vertices of $C$.)
As we will see, \textsc{Vertex Cover} is FPT parameterized by this parameter (even if a maximum clique is not given).
This gives us a theoretical reasoning as to why \textsc{Vertex Cover} is indeed trivial to solve on complete graphs.

We also believe that above structural guarantee parameterizations are of theoretical interest, because they are closely related to identifying graphs in which two of its parameters coincide.
Structural characterizations of such graphs have been extensively studied where the two parameters are maximum matching size and minimum vertex cover size \cite{LP86}, maximum matching size and minimum edge dominating set size \cite{LP1984}, maximum matching size and induced matching size \cite{DBLP:journals/dm/CameronW05a}, maximum independent set size and minimum dominating set \cite{Plu1970,Plu1993}, and minimum dominating set and minimum independent dominating set \cite{DBLP:journals/dm/GoddardH13}.
The corresponding computational complexity questions, i.e.\ whether these graphs can be recognized in polynomial time, have been studied as well~\cite{CHVATAL1993179,DBLP:journals/ipl/DemangeE14,DBLP:journals/disopt/DenizEHMS17,DBLP:journals/tcs/DuarteJPRS15,Gavril1997,DBLP:journals/networks/SankaranarayanaS92}.

Finally, our parameterizations can be seen as ``dual'' of parameters studied in literature.
There has been significant work (especially in the context of kernelization) on \textsc{Vertex Cover} parameterized by structural parameters smaller than the solution size~$k$ such as feedback vertex set number~\cite{DBLP:conf/icalp/BougeretJS20,DBLP:journals/mst/CyganLPPS14,DBLP:conf/stacs/HolsKP20,DBLP:journals/mst/JansenB13,DBLP:journals/mst/Majumdar0018}.

\paragraph*{Our contribution.}

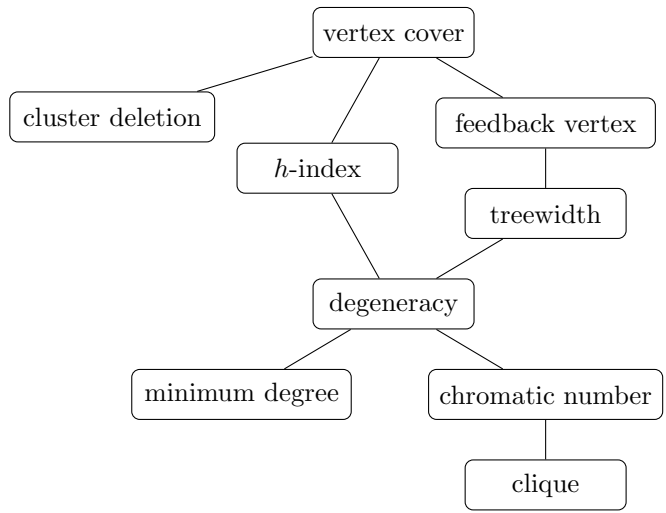
\begin{figure}
	\centering
	\begin{tikzpicture}[yscale=.8]
		\tikzset{
			param/.style={draw, fill=white, rectangle, rounded corners=3, font=\small, minimum width=5.5em, minimum height=4ex},
		}
		\node[param, inner sep=0pt,inner ysep=+0pt] at (0, 1.5) (vc) {vertex cover};
		\node[param, inner sep=0pt,inner ysep=+0pt] at (-1, -.75) (hi) {$h$-index};
		\node[param, inner sep=0pt,inner ysep=+0pt,minimum width=7.5em] at (2, 0) (fv) {feedback vertex};
		\node[param, inner sep=0pt,inner ysep=+0pt] at (2, -1.5) (tw) {treewidth};
		\node[param, inner sep=0pt,inner ysep=+0pt] at (0, -3) (dg) {degeneracy};
		\node[param, inner sep=0pt,inner ysep=+0pt,minimum width=8em] at (2, -4.5) (ch) {chromatic number};
		\node[param, inner sep=0pt,inner ysep=+0pt] at (2, -6) (cl) {clique};
		\node[param, inner sep=0pt,inner ysep=+0pt,minimum width=7.5em] at (-2, -4.5) (md) {minimum degree};
		\node[param, inner sep=0pt,inner ysep=+0pt,minimum width=7em] at (-3.7, 0.1) (cv) {cluster deletion};
		
		\draw (vc) -- (hi) -- (dg) -- (md);
		\draw (vc) -- (fv) -- (tw) -- (dg) -- (ch) -- (cl);
		\draw (vc) -- (cv);
	\end{tikzpicture}
	\caption{A Hasse diagram of graph parameters. There is a line between two parameters $p$ (above) and $q$ (below) if $p + 1 \ge q$ holds for any graph $G$.}
	\label{fig:hier}
\end{figure}
In \Cref{sec:vc-above-h}, we show that \textsc{Vertex Cover} is FPT when parameterized by $k - h$ for the $h$-index $h$.
This parameter is greater than or equal to many graph parameters such as degeneracy $d$, chromatic number $\chi$, and clique number $\omega$ (See \Cref{fig:hier}).
Thus, \textsc{Vertex Cover} is FPT for $k - d$, $k - \chi$, and $k - \omega$ as well.
Using a similar approach, we show in \Cref{sec:fvs-above-dg} that \textsc{Feedback Vertex Set} is FPT for $k - d$.
We also show that on planar graphs, fixed-parameter tractability of \textsc{Vertex Cover} with respect to $k - d$ can be strengthened: \textsc{Vertex Cover} is FPT parameterized by~$k - \tw$ for the treewidth $\tw$ (\Cref{sec:vc-above-tw}).
In the remaining sections, we prove hardness results.
In \Cref{sec:vc-above-cl}, we show that \textsc{Vertex Cover} admits no kernel of size polynomial in $k - \delta$  ($\delta$ is the minimum degree) and neither \textsc{Vertex Cover} nor \textsc{Feedback Vertex Set} admit a kernel of size polynomial in $k - \omega$.
We also show that \textsc{Vertex Cover} is \Wone-hard for $k - \fvn$ (\cref{sec:vc-above-fv}) and NP-hard for $k - \cvd = 0$ (\Cref{sec:vc-above-cd}), where $\fvn$ and $\cvd$ are the size of a minimum feedback vertex set and of a minimum cluster deletion set, respectively.
Finally, we prove that \textsc{Feedback Vertex Set} NP-hard for $\vc-k=2$ in \Cref{sec:fvs-below-vc} where $\vc$ is the size of a minimum vertex cover.

\section{Preliminaries}
\paragraph{Graphs.}
For standard graph terminology, we refer to Diestel~\cite{Diestel2017}.
All graphs we consider are finite, undirected, and loopless.
We call a function $p$ that maps any graph~$G$ to an integer $p(G)$ a \emph{graph parameter}.
In the following, we will define several graph parameters that are of interest in this work.
Let $G$ be a graph.
We denote the vertex set and edge set of $G$ by $V(G)$ and $E(G)$, respectively.
We denote the \emph{minimum degree} of $G$ by $\delta(G)$ and the \emph{maximum degree} by $\Delta(G)$.
The \emph{vertex cover number} $\vc(G)$ of $G$ is the size of a smallest set $X\subseteq V(G)$ such that~$G-X$ is edgeless.
The \emph{feedback vertex number}~$\fvs(G)$ of $G$ is the size of a smallest set~$X\subseteq V(G)$ such that $G-X$ is acyclic.
The \emph{$h$-index} $h(G)$ of a graph $G$ is the largest integer $k$ such that $G$ contains at least $k$ vertices each of degree at least $k$.
The \emph{degeneracy} of a graph $G$ is $d(G) \coloneqq \max_{V' \subseteq V(G)} \delta(G[V'])$.
A subset $V'$ of~$V(G)$ that maximizes $\delta(G[V'])$ is a \emph{core} of $G$.
The \emph{clique number}~$\omega(G)$ of~$G$ is the size of a largest clique in~$G$.
The \emph{chromatic number} $\chi(G)$ is the minimum integer $k$ such that $G$ can be properly $k$-colored.
The \emph{cluster deletion number}~$\cvd(G)$ of~$G$ is the size of a smallest set $X\subseteq V(G)$ such that $G-X$ does not contain a $P_3$ as an induced subgraph.
A pair $(T=(W,F),\beta)$ where $T$ is a tree and $\beta \colon W \rightarrow 2^{V(G)}$ is a \emph{tree decomposition} of $G$ if
\begin{inparaenum}[(i)]
	\item $\bigcup_{w \in W} \beta(w) = V(G)$,
	\item $\{ w \in W \mid v \in \beta(w) \}$ induces a connected subgraph of $T$ for all $v \in V(G)$ and
	\item for all $\{u,v\} \in E(G)$, there exists a $w\in W$ with $u,v \in \beta(w)$.
\end{inparaenum}
The \emph{width} of $(T=(W,F),\beta)$ is $\max_{w\in W} |\beta(w)|-1$.
The \emph{treewidth} $\tw(G)$ of $G$ is the minimum width over all tree decompositions of $G$.

If $p$ and $q$ are graph parameters, then we will say that $p$ is \emph{smaller} than~$q$ (and write~$p\preceq q$), if there is a constant $c$ such that $p(G) \leq q(G) + c$ for all graphs $G$.
This differs from the way the boundedness relation between graph parameters is usually defined~\cite{Jansen2013,Sorge2019}, but this stricter definition is necessary in the context of difference parameterizations.
This is because with this stricter definition the following is true (and easy to prove):
If $p,q,r$ are graph parameters such that $p \preceq q \preceq r$, then $r -q \preceq r- p$.
\Cref{fig:hier} depicts the graph parameters relevant to this work and the relationships between them.

\paragraph{Parameterized complexity.}
A \emph{parameterized problem} is a pair $(L,\kappa)$ where $L\subseteq \Sigma^*$ for a finite alphabet $\Sigma$ and $\kappa \colon \Sigma^* \rightarrow \NN$ is the parameter.
The problem is \emph{fixed-parameter tractable} (FPT) if it can be decided by an algorithm with running time $\bigO(f(\kappa(I)) \cdot |I|^c)$ where $I\in \Sigma^*$, $f$ is a computable function and $c$ is a constant.
Note that if $(L,\kappa)$ is FPT and $\kappa \preceq \kappa'$, then $(L,\kappa')$ is also FPT.
A \emph{kernel} for this problem is a polynomial-time algorithm that takes the instance $I$ and outputs a second instance $I'$ such that
\begin{inparaenum}[(i)]
	\item $I \in (L, \kappa) \iff I' \in (L, \kappa)$ and
	\item $|I'| \leq f(\kappa(I))$ for a computable function $f$.
\end{inparaenum}
The \emph{size} of the kernel is $f$.
There is a hierarchy of computational complexity classes for parameterized problems: $\mathrm{FPT} \subseteq \Wone \subseteq \cdots \subseteq \mathrm{XP}$.
To show that a parameterized problem~$(L,\kappa)$ is (presumably) not FPT one may use a \emph{parameterized reduction} from a~$\Wone$-hard problem to~$L$.
A parameterized reduction from a parameterized problem~$(L,\kappa)$ to another parameterized problem~$(L',\kappa')$ is a function that acts as follows:
For computable functions~$f$ and~$g$, given an instance~$I$ of~$L$, it computes in $f(\kappa(I)) \cdot |I|^{O(1)}$ time an instance~$I'$ of~$L'$ so that~$I \in (L,\kappa) \iff I' \in (L',\kappa')$ and~$\kappa(I') \le g(\kappa(I))$.
For more details on parameterized algorithms and complexity, we refer to the standard literature~\cite{bluebook,Downey2013,Flum2006}.

\paragraph{Problem definitions.}
We are interested in above guarantee parameterizations of \textsc{Vertex Cover} of the following form.
Let $p \preceq \vc$ be a graph parameter.
Then, we define:
\boxproblem{Vertex Cover above $p$}
{A graph~$G$ and an integer~$k$.}
{Does~$G$ contain a vertex cover of order at most~$k$?}
{$\ell \coloneqq k- p(G)$.}

Similarly, we also consider above (below) guarantee parameterizations of \textsc{Feedback Vertex Set}.
Now, let $p$ be a graph parameter with $p \preceq \fvs$ ($\fvs \preceq p$).
We consider the following problem:
\boxproblem{Feedback Vertex Set above (below) $p$}
{A graph~$G$ and an integer~$k$.}
{Does~$G$ contain a feedback vertex set of order at most~$k$?}
{$\ell \coloneqq k- p(G)$ ($\ell \coloneqq p(G) - k$).}

\section{Vertex Cover above \texorpdfstring{$h$}{h}-Index}
\label{sec:vc-above-h}
We start by proving that \textsc{Vertex Cover} is \FPT{} when parameterized by the difference between $k$ and the $h$-index of the graph.
Recall that the state-of-the-art algorithm for \textsc{Vertex Cover} parameterized by the solution size~$k$ has running time $\bigO(1.274^k + kn)$~\cite{DBLP:journals/tcs/ChenKX10}. 

\begin{theorem}
	\textsc{Vertex Cover above $h$-Index} is \FPT.
\end{theorem}
\begin{proof}
	Let $(G,k)$ be an instance of \textsc{Vertex Cover} where $G$ is a graph with an $h$-index of $h$.
	Let $v_1,\ldots,v_h \in V(G)$ with $\deg(v_i) \geq h$.
	We branch into the following $h+1$ cases:
	\begin{enumerate}[(1)]
		\item The solution contains all of the vertices $v_1,\ldots,v_h$.
		Hence, we test the instance $(G-\{v_1,\ldots,v_h\},k-h)$ in time $\bigO(1.274^{k-h} + (k-h)n)$.
		\item The solution does not contain $v_i$ for some $i\in\{1,\ldots,h\}$.
		Then, the solution must contain all of $v_i$'s neighbors.
		Hence, we test the instance $(G - N(v_i),k-|N(v_i)|)$.
		Since $|N(v_i)| \geq h$, this is possible in time $\bigO(1.274^{k-h} + (k-h)n)$.
	\end{enumerate}
	In all, we get a running time of $\bigO(1.274^{k-h}h + (k-h)hn)$.
\end{proof}

This algorithm can also be used to obtain a Turing kernelization (cf.~\cite[Ch.~22]{Fomin2019}) by simply computing a kernel for each of the $h+1$ instances of \textsc{Vertex Cover} parameterized by the solution size that the algorithm branches into.

\section{Feedback Vertex Set above Degeneracy}
\label{sec:fvs-above-dg}
A similar approach to the one used in the previous section, branching once to lower $k$ and then applying a known algorithm for the standard parameterization, can also be employed to show that \textsc{Feedback Vertex Set above Degeneracy} is \FPT{}.
The fastest presently known deterministic algorithm for the standard parameterization of \textsc{Feedback Vertex Set} runs in time~$\bigO(3.460^k \cdot n)$~\cite{Iwata2021}.

\begin{theorem}
	\textsc{Feedback Vertex Set above Degeneracy} is \FPT{}.
\end{theorem}
\begin{proof}
	Let $(G,k)$ be an instance for \textsc{Feedback Vertex Set} where $d$ is the degeneracy of $G$.
	It is well-known that the degeneracy and the core of a graph can be computed in polynomial time by iteratively deleting a minimum-degree vertex and storing the largest degree of a vertex at the time it is deleted.
	We start by computing a core $V'$ of $G$.
	We branch into the following $|V'|^2 + 1$ cases.
	\begin{enumerate}[(1)]
		\item The entire core is contained in the minimum feedback vertex set.
		The core must contain at least $d+1$ vertices.
		Hence, we test the instance $(G-V',k-|V'|)$ in time $\bigO(3.460^{k-|V'|} \cdot n) = \bigO(3.460^{k-d} \cdot n)$.
		\item The entire core is not contained in the minimum feedback vertex set.
		Let $X$ denote the minimum feedback vertex and let $F := V(G) \setminus X$ be the maximum induced forest.
		\begin{enumerate}
			\item If $G[V'\cap F]$ contains an isolated vertex $u$, then all neighbors of~$u$ in~$G[V']$, of which there are at least $d$, must be in $X$.
			Hence, for each $u \in V'$, we test the instance $(G-N_{G[V']}(u),k-\deg_{G[V']}(u))$ in time $\bigO(3.460^{k-\deg_{G[V']}(u)} \cdot n) = \bigO(3.460^{k-d} \cdot n)$.
			\item If $G[V'\cap F]$ does not contain an isolated vertex, it must still contain a leaf $u$, since it is a forest.
			Then, all but one of the neighbors of~$u$ in~$G[V']$ must be in $X$.
			Hence, for each pair $u\in V'$ and~$v \in N_{G[V']}(u)$, we test the instance $(G-(N_{G[V']}(u) \setminus\{v\}),k-\deg_{G[V']}(u) + \nolinebreak 1)$ in time $\bigO(3.460^{k-\deg_{G[V']}(u)+1} \cdot n) = \bigO(3.460^{k-d} \cdot n)$.
		\end{enumerate}
	\end{enumerate}
	In all, this makes for a running time of $\bigO(3.460^{k-d}\cdot n^3)$.
\end{proof}

Like the algorithm in the previous section, this one can also be easily converted to a Turing kernelization.

\section{Vertex Cover above Treewidth}
\label{sec:vc-above-tw}

In this section, we show that on planar graphs \textsc{Vertex Cover} is also FPT with respect to $\ell = k - \tw$, which is smaller than $k - d$.

\begin{theorem}
	\textsc{Vertex Cover above Treewidth} on planar graphs is \FPT.
\end{theorem}
\begin{proof}
	Given a planar graph $G$, we compute the branchwidth $\beta$ of~$G$.
	This is possible in polynomial time, because $G$ is planar~\cite{Seymour1994}.
	Moreover,~$\beta \leq \tw(G) + 1 \leq \frac{3}{2} \beta$~\cite[Theorem~5.1]{Robertson1991}.
	Any planar graph with treewidth $w$ contains a $g\times g$-grid with $g\geq\frac{w +4}{6}$ as a minor~\cite[Theorem~6.2]{Robertson1994}.
	Hence, having computed~$\beta$, we know that~$G$ must contain a $g\times g$-grid with $g\geq \frac{\beta+3}{6}$ as a minor.
	Any vertex cover of the $g\times g$-grid must contain at least $\lfloor\frac{g}{2}\rfloor$ in each row, for a total of at least $g \cdot \lfloor\frac{g}{2}\rfloor \geq  \frac{g(g-1)}{2}$ vertices.
	Since~$\vc(H_1) \leq \vc(H_2)$, if~$H_1$ is a minor of $H_2$, it follows that~$\vc(G) \geq \frac{\beta^2-9}{72}\eqqcolon r$.
	Hence, if~$k < r$, we may reject the input.
	Otherwise, $\ell = k - \tw(G) \geq k - \frac{3}{2} \beta \geq k - \frac{3}{2}\sqrt{72k - 9}$.
	This means that $\ell$ is bounded from below by a function in $k$ and, therefore, fixed-parameter tractability with respect to $k$ implies fixed-parameter tractability with respect to $\ell$.
\end{proof}

This algorithm relies on two properties of planar graphs:
\begin{inparaenum}[(i)]
	\item large treewidth guarantees the existence of a $g\times g$-grid where $g \in \Omega(\tw^{1/2 + \varepsilon})$ for $\varepsilon > 0$ and
	\item branchwidth can be computed in polynomial time on planar graphs.
\end{inparaenum}
In any graph class that excludes a minor, (i)~still holds true~\cite{Demaine2008}.
Although it is not clear that (ii) can be generalized, we remark that a constant approximation algorithm is known for graphs excluding single-crossing graphs as minors~\cite{Demaine2004}.
In fact, our result can be extended to any class of graphs that do not contain a single-crossing graph as a minor.

We leave open whether or not \textsc{Vertex Cover above Treewidth} is FPT on graph classes that exclude a minor (other than planar graphs) or even on arbitrary graphs.

\section{Kernelization Lower Bounds}
\label{sec:nopk}
\label{sec:vc-above-cl}

In this section we show that, while there is a Turing kernel when parameterized above $h$-index, \textsc{Vertex Cover} presumably does not admit a polynomial kernel when parameterized above the minimum degree or the clique number.

\begin{theorem}
	\label{thm:vc-above-clique-nopk}
	\textsc{Vertex Cover above Minimum Degree} and \textsc{Vertex Cover above Clique Number} do not admit a polynomial kernel unless \NPincoNPslashpoly{}.
\end{theorem}
\begin{proof}
	We prove the statement by giving a linear parametric transformation from \textsc{Clique} parameterized by maximum degree and parameterized by the vertex cover number.
	Unless \NPincoNPslashpoly{}, under neither parameterization does \textsc{Clique} admit a polynomial kernel.
	This is folklore for maximum degree and was shown by Bodlaender~et~al.~\cite{Bodlaender2014} for vertex cover number.
	The underlying reduction takes the \textsc{Clique} instance $(G, k)$ and transforms it into the instance~$(\bar G, \bar k)$ of \textsc{Vertex Cover} where $\bar G$ is the complement graph of~$G$, that is $V(\bar G) \coloneqq V(G)$ and $E(\bar G) \coloneqq \binom{V(G)}{2}\setminus E(G)$, and $\bar k \coloneqq \abs{V}-k$.
	The reduction is obviously correct and computable in~$\bigO(\abs{V}^2)$ time.
	As for the parameterizations, observe that~$\bar k - \delta(\bar G) = (\abs{V} - k) - (\abs{V} - 1 - \Delta(G)) \le \Delta(G)$.
	Since~$\omega(\bar G) \ge |V(G)| - \vc(G)$, we also have $\bar k - \omega(\bar G) \le (|V(G)| - k) - (n - \vc(G)) \le \vc(G)$. 
	This yields the claimed transformations.
\end{proof}

Using a standard reduction from \textsc{Vertex Cover} to \textsc{Feedback Vertex Set}, we obtain the following.
\begin{corollary}
	\label{cor:fvs-above-clique-nopk}
	\textsc{Feedback Vertex Set above Clique Number} does not admit a polynomial kernel unless \NPincoNPslashpoly.
\end{corollary}
\begin{proof}
	We provide a linear parametric transformation from \textsc{Vertex Cover above Clique Number}.
	Given an instance~$(G, \ell)$ of \textsc{Vertex Cover}, we use the following folklore construction to obtain an instance~$(G', \ell)$ of \textsc{Feedback Vertex Set}.
	After initializing~$G'$ as a copy of~$G$, we add for each edge $\{u, v\} = e \in E(G)$ the vertex~$w_{uv}$ and the edges~$\{u, w_{uv}\}$ and~$\{v, v_{uv}\}$ to~$G'$, so that for each edge~$e \in E(G)$ there exists a unique triangle in~$G'$.
	Clearly, unless~$\omega(G) = 2$, we have~$\omega(G') = \omega(G)$.
	Hence, the parameter~$\ell - \omega(G')$ is upper-bounded by the parameter of the input problem, and we are done.
\end{proof}

We leave open whether \textsc{Feedback Vertex Set above Minimum Degree} admits a polynomial kernel.

\section{Vertex Cover above Feedback Vertex Number}
\label{sec:vc-above-fv}
\newcommand{\VCfvn}{\prob{Vertex Cover above Feedback Vertex Number}}

In this section we prove that, when parameterizing above feedback vertex number, \textsc{Vertex Cover} is \Wone-hard.
First, we prove \Wone-hardness with respect to a related parameter, namely above the \emph{distance to $K_r$-free} for every constant~$r \ge 3$, that is, the minimum number of vertices one needs to remove such that the remaining graph does not contain a clique of order $r$.
Distance to $K_3$-free is a lower bound on the feedback vertex number, so our proof also implies hardness for \VCfvn.

\begin{theorem}
	\label{thm:vc-dist-to-kr-free}
	\textsc{Vertex Cover above Distance to $K_r$-Free} is \Wone-hard.
\end{theorem}
\begin{proof}
	We provide a parameterized reduction from \prob{Independent Set} parameterized by the solution size~$k$.
	Let~$I = (G, k)$ be an instance of \prob{Independent Set} with~$V(G) = \oneto{n}$ and~$E(G) = \{e_1, \dots, e_m\}$.
	We create an instance~$(G', \ell)$ as follows.
	First, for each~$i \in \oneto{k}$ we add a clique on the vertex set~$V_i =\{w^i_j \mid j \in \oneto{n}\}$ to~$G'$.
	For each~$i, j \in \oneto{k}$, we add the edge between~$w^i_q$ and~$w^j_q$ for each~$q \in [n]$ and the edge between~$w^i_p$ and~$w^i_q$ for each~$\{p,q\} \in E$.
	Next, for each~$i \in \oneto{k}$ we add a set~$A_i$ of $r-2$ vertices which form a clique, attach a leaf to each vertex in~$A_i$, and make each vertex in~$A_i$ adjacent to each vertex in~$V_i$.
	Let~$A \coloneqq \bigcup_{i \in \oneto{k}} A_i$.
	Then we add $k+1$ cliques on~$r-1$ vertices.
	Call the set of these vertices~$B$, attach a leaf to each~$v \in B$, and make each~$v \in B$ adjacent to each vertex in~$\bigcup_{i \in \oneto{k}} V_i$.
	Denote by~$L$ the set of leaves in~$G'$.
	Lastly, set $\ell \coloneqq (n-1)k + \abs{A} + \abs{B} = (n-1)k + |L|$.

	For the correctness, observe that $G'$ contains a vertex cover of size~$k'$ if and only if~$G'$ contains an independent set of size
	\[\abs{V'}-\ell = nk + \abs{A} + \abs{B} + \abs{L} - ((n-1)k + \abs{A} + \abs{B}) = \abs{L}+k.\]
	Let~$Y \subseteq V(G')$ be an independent set in~$G'$.
	As it is always optimal to take leaves into an independent set, we may assume that~$Y$ contains all of~$L$.
	Hence, we may assume that~$Y \cap (A \cup B) = \emptyset$.
	Furthermore, $Y$ contains at most one vertex of each clique on~$V_i$, $i \in \oneto{k}$, and, by the construction of the edges in~$G'$, $Y$ can contain only such vertices in the cliques whose corresponding vertices in~$G$ are pairwise nonadjacent.
	Hence, $G'$ contains an independent set of size~$\abs{L}+k$ if and only if~$G$ contains an independent set of size~$k$, and the reduction is correct.

	Finally, we will show that the distance to~$K_r$-free of $G'$ is exactly~$nk$.
	This implies that the parameter of the output instance, if $d$ is the distance to $K_r$-free of $G'$, is $k'= \ell - d = (n-1)k + \abs{A} + \abs{B} - nk = \abs{A} + \abs{B} - k = (r-2)k + (r-1)(k+1) - k$.
	Since $r$ is a constant, this implies that $k'$ is bounded in $k$.
	Let~$D \subseteq V'$ be of minimum set such that~$G'-D$ is~$K_r$-free.
	Clearly, $D \cap L = \emptyset$ as~$L$ does not intersect any~$K_r$.
	Furthermore, as every vertex~$u \in A_i$ intersects a subset of the cliques that any vertex~$v \in V_i$ intersects, we may exchange each vertex in~$D \cap A_i$ with a vertex in~$V_i$.
	As there are fewer than $r$ vertices in each~$A_i$ we may assume that~$D \cap A_i = \emptyset$.
	But then $D$ must contain $n-1$ vertices of each set~$V_i$, hence, $D$ contains all but~$k$ vertices from~$\bigcup_{i \in \oneto{k}} V_i$.
	If, however, $v \notin D$ for one such~$v \in \bigcup_{i \in \oneto{k}} V_i$, then~$D$ must contain at least one vertex from each clique in~$B$.
	As there are $k+1$ such cliques, $v \notin D$ contradicts $D$ being minimum.
	Hence, $D \coloneqq \bigcup_{i \in \oneto{k}} V_i$ is a $K_r$-deletion set of size~$nk$.
\end{proof}

For $r=3$, the deletion set~$D$ in the proof above is also a feedback vertex set.
Hence, we obtain the following.
\begin{corollary}
	\label{cor:vc-fvn}
	\textsc{Vertex Cover above Feedback Vertex Number} is \Wone-hard.
\end{corollary}

Observe that in the proof of \cref{thm:vc-dist-to-kr-free} we can specify a minimum deletion set.
Hence, our hardness results also hold if a minimum deletion set is given as part of the input.

\section{Vertex Cover above Cluster Deletion Number}
\label{sec:vc-above-cd}

Recall that the cluster deletion number is the minimum size of a set $X$ such that $G - X$ is a cluster graph, i.e., every connected component of $G - X$ is a clique.
We show that \textsc{Vertex Cover Above Cluster Deletion Number} is NP-hard even if the parameter is zero.

\begin{theorem}
	\label{thm:cvd}
	\textsc{Vertex Cover above Cluster Deletion} is NP-hard even if $\ell = 0$.
\end{theorem}

We will prove this theorem by reduction from \textsc{3-SAT}.
In fact, we prove a slightly stronger claim: \textsc{Vertex Cover} is NP-hard when restricted to graphs $G$ with $\cvd(G) = \vc(G)$.
The following construction is illustrated in \Cref{fig:above-cluster-deletion}.

\begin{figure}
	\centering
	\begin{tikzpicture}[scale=1.5]
		\tikzstyle{xnode}=[circle,scale=0.3,draw,fill=white];
		\node[xnode] (u11) at (-1,0.5) {};
		\node[below=0.1cm of u11] {\tiny $v_{2,1,1}$};
		\node[xnode] (u12) at (0.5,0.5) {};
		\node[below=0.1cm of u12] {\tiny $v_{2,2,1}$};
		\node[xnode] (u13) at (2,0.5) {};
		\node[below=0.1cm of u13] {\tiny $v_{2,3,1}$};
		\node[xnode] (u21) at (-1,1) {};
		\node[left=0.005cm of u21] {\tiny $v_{1,1,1}$};
		\node[xnode] (u22) at (0.5,1) {};
		\node[left=0.005cm of u22] {\tiny $v_{1,2,1}$};
		\node[xnode] (u23) at (2,1) {};
		\node[left=0.005cm of u23] {\tiny $v_{1,3,1}$};
		\foreach \x in {1,2,3} {
			\foreach \y in {1,2,3} {
				\draw (u1\x) -- (u2\y);	
			}
		}
	
		\node[xnode] (v11) at (3,0.5) {};
		\node[below=0.1cm of v11] {\tiny $v_{1,1,2}$};
		\node[xnode] (v12) at (4.5,0.5) {};
		\node[below=0.1cm of v12] {\tiny $v_{1,2,2}$};
		\node[xnode] (v13) at (6,0.5) {};
		\node[below=0.1cm of v13] {\tiny $v_{1,3,2}$};
		\node[xnode] (v21) at (3,1) {};
		\node[left=0.005cm of v21] {\tiny $v_{2,1,2}$};
		\node[xnode] (v22) at (4.5,1) {};
		\node[left=0.005cm of v22] {\tiny $v_{2,2,2}$};
		\node[xnode] (v23) at (6,1) {};
		\node[left=0.005cm of v23] {\tiny $v_{2,3,2}$};
		\foreach \x in {1,2,3} {
			\foreach \y in {1,2,3} {
				\draw (v1\x) -- (v2\y);	
			}
		}
		
			\node[xnode] (w11) at (1,4) {};
			\node[above=0.005cm of w11] {\tiny $u_{3,1,1}$};
			\node[xnode] (w12) at (1.5,4) {};
			\node[above=0.005cm of w12] {\tiny $u_{3,2,1}$};
			\node[xnode] (w13) at (2,4) {};
			\node[above=0.005cm of w13] {\tiny $u_{3,3,1}$};
			\node[xnode] (w14) at (2.5,4) {};
			\node[above=0.005cm of w14] {\tiny $u_{3,4,1}$};
			\node[xnode] (w15) at (3,4) {};
			\node[above=0.005cm of w15] {\tiny $u_{3,5,1}$};
			\node[xnode] (w16) at (3.5,4) {};
			\node[above=0.005cm of w16] {\tiny $u_{3,6,1}$};
			\node[xnode] (w17) at (4,4) {};
			\node[above=0.005cm of w17] {\tiny $u_{3,7,1}$};
			
			\node[xnode] (w21) at (1,3.25) {};
			\node[xnode] (w22) at (1.2,3) {};
			\node[xnode] (w23) at (1.4,2.75) {};
			\node[xnode] (w24) at (1.6,2.5) {};
			\node[xnode] (w25) at (1.8,2.25) {};
			\node[xnode] (w26) at (2,2) {};
			\node[xnode] (w27) at (2.2,1.75) {};
			
			\node[xnode] (w31) at (4,3.25) {};
			\node[xnode] (w32) at (3.8,3) {};
			\node[xnode] (w33) at (3.6,2.75) {};
			\node[xnode] (w34) at (3.4,2.5) {};
			\node[xnode] (w35) at (3.2,2.25) {};
			\node[xnode] (w36) at (3,2) {};
			\node[xnode] (w37) at (2.8,1.75) {};
		
		\foreach \x in {1,...,7} {
			\foreach \y in {1,...,7} {
				\foreach \z in {1,...,7} {
					\draw[gray,line width=0.025pt] (w1\x) -- (w2\y);
					\draw[gray,line width=0.025pt] (w1\x) -- (w3\z);
					\draw[gray,line width=0.025pt] (w2\y) -- (w3\z);
				}
			}
		}

		\foreach \x in {1,2,3} {
			\foreach \y in {1,...,7} {
				\draw (u2\x) -- (w2\y);
			}
		}
	
		\foreach \x in {1,2,3} {
			\foreach \y in {1,...,7} {
				\draw (v2\x) -- (w3\y);
			}
		}
	
		\node[left=of u21] {$x_1$};
		\node[left=of u11] {$\neg x_1$};
		\node[right=of v23] {$\neg x_2$};
		\node[right=of v13] {$x_2$};
	\end{tikzpicture}
	\caption{An excerpt of the graph $G$ output by
	\cref{constr:cvd}: At the bottom are the vertex gadgets for the variables $x_1$ and $x_2$.
	At the top is a clause gadget representing a clause that contains both $x_1$ and $\neg x_2$.}
	\label{fig:above-cluster-deletion}
\end{figure}
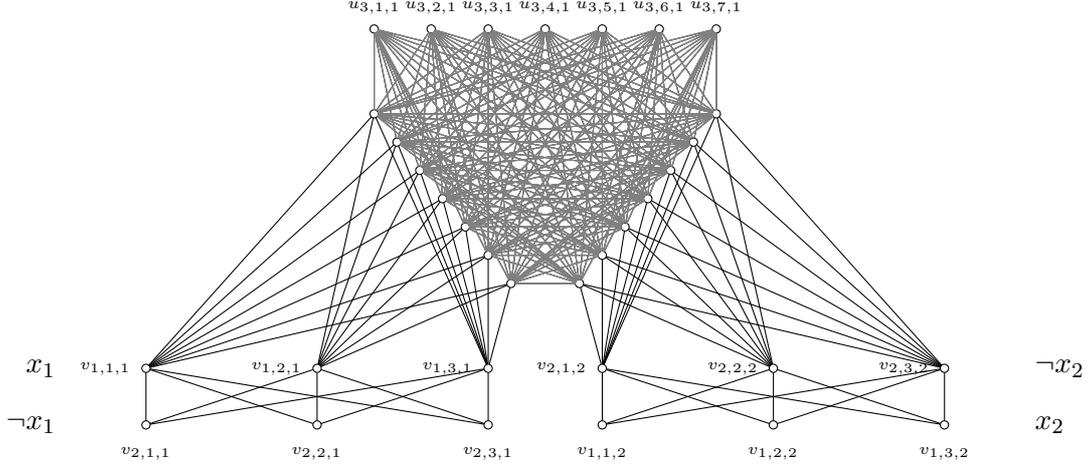

\begin{construction}
	\label{constr:cvd}
	Let $\varphi$ be a Boolean formula in 3-CNF consisting of the clauses $C_1,\ldots,C_m$ over the variables~$x_1,\ldots,x_n$.
	We may assume that each clause of $\varphi$ contains exactly three literals.
	For each $i \in \{1,\ldots,m\}$, let~$C_i = (L_i^1 \vee L_i^2 \vee L_i^3)$ where $L_i^1$, $L_i^2$, and $L_i^3$ are literals.
	
	We construct a graph $G$ and an integer $k \coloneqq 14m + 3n$ such that~$\cvd(G) = \vc(G)$ and $\vc(G) \leq k$ if and only if $\varphi$ is satisfiable.
	Each variable $x_j$ is represented by a \emph{variable gadget} consisting of six vertices $(v_{r,s,j})_{r\in\{1,2\},s\in\{1,2,3\}}$ that induce a complete bipartite graph with three vertices in each color class.
	Each clause $C_i$ is represented by a \emph{clause gadget} consisting of twenty-one vertices $(u_{r,s,i})_{r \in\{1,2,3\},s\in\{1,\ldots,7\}}$ that induce a complete tripartite graph with seven vertices in each color class.
	The two sides of the bipartition of a variable gadget correspond to its positive and negative literals and the three sides of the tripartition of a clause gadget correspond to the three literals the clause contains.
	All vertices in a side of a clause gadget are connected to all vertices in the side of a variable gadget if these two sides correspond to the literal of opposite sign.
	Formally, we let:
	\begin{alignat*}{2}
		V(G)& \coloneqq && \{ u_{r,s,i} \mid r \in\{1,2,3\},s\in\{1,\ldots,7\},i\in\{1,\ldots,m\}\}\\
		& && \cup \{v_{r,s,j} \mid r\in\{1,2\},s\in\{1,2,3\},j\in\{1,\ldots,n\}\} \text{ and}\\
		E(G)&\coloneqq&& \{\{u_{r,s,i},u_{r',s',i}\} \mid r,r' \in\{1,2,3\},r\neq r',s,s' \in\{1,\ldots,7\},i\in\{1,\ldots,m\}\}\\
		& &&\cup \{ \{v_{1,s,j},v_{2,s',j}\} \mid s,s'\in\{1,2,3\},j\in\{1,\ldots,n\}\}\\
		& && \cup \{\{u_{r,s,i},v_{1,s',j}\} \mid s\in\{1,\ldots,7\},s'\in\{1,2,3\},L^r_i = x_j\}\\
		& && \cup \{\{u_{r,s,i},v_{2,s',j}\} \mid s\in\{1,\ldots,7\},s'\in\{1,2,3\}, L^r_i = \neg x_j\}.
	\end{alignat*}
\end{construction}

\begin{lemma}
	\label{lemma:cvd1}
	Let $G$ be the graph output by \cref{constr:cvd} and $X\subseteq V(G)$ be a minimum cluster deletion set.
	Then, in each clause gadget of $G$, $X$ contains all vertices in two of the sides of the gadget and none of the vertices of the third side.
\end{lemma}
\begin{proof}
	For every $i \in \{1,\ldots,m\}$, the deletion set $X$ must contain either
	\begin{inparaenum}[(i)]
		\item all vertices from two sides of the clause gadget for $C_i$ or
		\item all but one vertex from each side of this clause gadget,
	\end{inparaenum}
	because if $X$ omits two vertices from one side and an additional vertex from a second side these three vertices induce a $P_3$.
	
	In case (i), it only remains to show that $X$ does not contain any of the vertices in the third side.
	Let $A \coloneqq \{u_{r,1,i},\ldots,u_{r,7,i}\}$ with $r \in \{1,2,3\}$ be the vertices in this side.
	Let $\alpha_r \coloneqq 1$, if $L^r_i = x_{j}$, and $\alpha_r \coloneqq 2$, if $L^r_i = \neg x_{j}$.
	If $X$ contains all of the vertices in $B\coloneqq \{v_{\alpha_r,1,j}, v_{\alpha_r, 2, j}, v_{\alpha_r,3,j}\}$, then
	$X \setminus A$ is a cluster deletion set strictly smaller than $X$.
	If $X$ does not contain all vertices in $B$, then it must contain all but one of the vertices in $A$ (i.e., $|X \cap A| \ge 6$).
	Then $(X \setminus A) \cup B$ is a cluster deletion set and $\abs{(X \setminus A) \cup B} \leq \abs{X} - 6 + 3 < \abs{X}$.
	
	For case (ii), let $r, r', r''$ be the three sides of the clause gadget and let $X_{r'''}$ be the set of vertices in $X$ that lie in side $r'''$.
	Assume that $X$ does not contain all vertices from two sides~$r,r'$.
	Since $X$ must contain all but one vertex from each side, we may assume without loss of generality that $X_r = \{ u_{r,1,i}, \dots, u_{r,6,i} \}, |X_{r'}| = \{ u_{r',1,i}, \dots, u_{r',6,i} \}$, and $|X_{r''}| \ge 6$ (note that it may contain all vertices in the third side $r''$).
	Let $\alpha_{r''} \coloneqq 1$, if $L^{r''}_i = x_{j}$, and $\alpha_{r''} \coloneqq 2$, if $L^{r''}_i = \neg x_{j}$ and let
	\begin{align*}
		X' \coloneqq &\, (X \setminus X_{r''})
		\cup \{u_{r,7,i}, u_{r',7,i} \} \cup \{v_{\alpha_{r''},s,j} \mid s \in \{1,2,3\}\}.
	\end{align*}
	Then, $|X'| \leq |X| - 6 + 2 + 3 = |X| -1$.
	Moreover, $X'$ is a cluster deletion set in $G$, because $X \setminus X' \subseteq  \{u_{r'',s,i} \mid s\in\{1,\ldots,7\}\}$, but all of these vertices are isolated and, therefore, not part of any $P_3$ in $G-X'$.
	Hence, $X$ is not minimum.
\end{proof}

A similar statement is also true for vertex gadgets:

\begin{lemma}
	\label{lemma:cvd2}
	Let $G$ be the graph output by \cref{constr:cvd} and $X\subseteq V(G)$ be a minimum cluster deletion set.
	Then, in each variable gadget of $G$, $X$ contains all vertices in one of the two sides of the gadget.
\end{lemma}
\begin{proof}
	If $X$ does not contain all vertices in either side of a vertex gadget corresponding to the variable $x_j$, it must contain all but one vertex from each side.
	Let $C_{i_1},\ldots,C_{i_t}$ be the clauses that contain the literal $x_j$ and let $r_{1},\ldots,r_{t}$ be the sides of each of the corresponding clause gadgets whose vertices are adjacent to the side in the vertex gadget of $x_j$.
	Then, $X$ must contain all but one vertex in the side $r_{t'}$ of the clause gadget for $C_{i_{t'}}$ for each~$t' \in \{1,\ldots,t\}$. 
	By \cref{lemma:cvd1}, it follows that $X$ contains all vertices in each of those sides.
	Hence, removing from $X$ all vertices in the $r=1$ side of $x_j$'s vertex gadget and adding the remaining vertex in the $r=2$ side yields a smaller cluster deletion set.
\end{proof}

\begin{lemma}
	\label{lemma:cvd3}
	Let $G$ be the graph output by \cref{constr:cvd} and $C\subseteq V(G)$ be a minimum cluster deletion set.
	Then, $C$ is a vertex cover of $G$.
	Hence, $\cvd(G) = \vc(G)$.
\end{lemma}
\begin{proof}
	Clearly $\cvd(G) \le \vc(G)$.
	We show that $\cvd(G) \ge \vc(G)$.
	By \cref{lemma:cvd1,lemma:cvd2}, $C$ covers all edges within each clause and each variable gadget.
	It remains to show that edges between these gadgets are covered.
	The only such edges are between sides of a clause gadget and sides of a variable gadget when these two sides correspond to the same literal.
	Then, $C$ must contain all vertices in one of these two sides, since, otherwise, $G-C$ would contain an induced $P_3$.
	Hence, $C$ also covers all edges between those two sides.
\end{proof}

\begin{lemma}
	\label{lemma:cvd4}
	Let $\varphi$ be a formula in 3-CNF and $G$ the graph output by \cref{constr:cvd} on input $\varphi$.
	Then, $\varphi$ is satisfiable if and only if $\vc(G) \leq \ell$.
\end{lemma}
\begin{proof}
	First, suppose that $\varphi$ is satisfiable and that $\alpha \colon \{x_1,\ldots,x_n\} \rightarrow \{0,1\}$ is a satisfying assignment.
	We extend $\alpha$ to literals on this variable set in the natural way.
	Since $\alpha$ satisfies every clause in $\varphi$, there is an $\alpha_i \in \{1,2,3\}$ for every $i\in\{1,\ldots,m\}$ such that $\alpha(L_i^{\alpha_i}) = 1$.
	Let
	\begin{align*}
		C \coloneqq \, & \{u_{r,s,i} \mid r \in \{1,2,3\} \setminus \{\alpha_i\},s\in\{1,\ldots,7\},i\in\{1,\ldots,m\}\}\\
		&\cup \{v_{1,s,j} \mid s\in\{1,2,3\},j\in\{1,\ldots,n\}, \alpha(x_j) = 1\}\\
		&\cup \{v_{2,s,j} \mid s\in\{1,2,3\},j\in\{1,\ldots,n\}, \alpha(x_j) = 0\}.
	\end{align*}
	First, note that $|C| = 14 m + 3n = k$.
	Secondly, we claim that $C$ is a vertex cover of $G$.
	Clearly, all edges within clause gadgets and all edges within vertex gadgets are covered, because $C$ contains all but one side in each of those gadgets.
	Edges between a vertex gadget and a clause gadget are covered because $C$ contains vertices in sides of vertex gadgets, unless the literal this side corresponds to is not satisfied by $\alpha$, but if this is the case, then $C$ contains all sides of clause gadgets that correspond to this literal.
	
	Now, suppose that $C\subseteq V(G)$, $|C|\leq k$, is a vertex cover of $G$.
	We may assume that $C$ is minimum.
	Hence, by \cref{lemma:cvd1,lemma:cvd2,lemma:cvd3}, it contains all vertices in at two sides of every clause gadget and all vertices in exactly one side of every variable gadget.
	Let $\alpha \colon \{x_1,\ldots,x_n\} \rightarrow \{0,1\}$ with:
	\begin{align*}
		\alpha(x_j) \coloneqq
		\begin{cases}
			1, & \text{if } \{v_{1,s,j} \mid s \in \{1,2,3\} \} \subseteq C,\\
			0, & \text{if } \{v_{2,s,j} \mid s \in \{1,2,3\} \} \subseteq C.
		\end{cases}
	\end{align*}
	We claim that $\alpha$ satisfies $\varphi$.
	Let $C_i$ be a clause in $\varphi$.
	One of the three sides of the gadget representing $C_i$ is not contained in $C$.
	This side corresponds to the literal $L^r_i \in \{x_j, \neg x_j\}$.
	If $L^r_i = x_j$, then all vertices in $\{u_{r,s,i} \mid s\in \{1,\ldots,7\}\}$ are adjacent to all vertices in $\{v_{1,s,j} \mid s \in \{1,2,3\}\}$.
	Since $\{u_{r,s,i} \mid s\in \{1,\ldots,7\}\} \not\subseteq C$, it follows that $\{v_{1,s,j} \mid s \in \{1,2,3\}\}\subseteq C$ and, therefore, $\alpha(x_j) = 1$.
	Hence, $\alpha$ satisfies the clause $C_i$.
	The case where $L^r_i = \neg x_j$ is analogous.
\end{proof}

\cref{thm:cvd} follows from the preceding lemmas.

\begin{proof}[Proof of \cref{thm:cvd}]
	Clearly, \cref{constr:cvd} can be computed in polynomial time.
	The claim follows by \cref{lemma:cvd3,lemma:cvd4}.
\end{proof}

We remark that the NP-hardness of \textsc{Vertex Cover above Cluster Deletion Number} holds even if we are given a minimum cluster deletion set as part of the input.
To show this, we slightly adapt \Cref{constr:cvd}.
Let $(G, k)$ be an instance given in \Cref{constr:cvd}.
We further introduce $7m/3 + n$ complete graphs on three vertices (we may assume that $m$ is divisible by $3$, otherwise we add dummy clauses). Denote these vertices by $T$ and observe that $\abs{T} = 7m + 3n$.
We add an edge between each vertex in $V(G)$ and each vertex in $T$.
Let $H$ denote the resulting graph.
Observe that $V(G)$ is a cluster vertex deletion set of size $21m + 6n$ of $H$.
Moreover, $\cvd(H) \geq \cvd(G) + \abs{T} \geq 21m + 6n$.
It is not difficult to show that any vertex cover of size at most $21m + 6n$ of $G$ contains every vertex of $T$.
Thus, $H$ has a vertex cover of size at most $21m + 6n$ if and only if $G$ has a vertex cover of size $21m + 6n - |T| = k$.
\Cref{lemma:cvd4} establishes the correctness of the reduction.

\section{Feedback Vertex Set below Vertex Cover}
\label{sec:fvs-below-vc}

\textsc{Feedback Vertex Set below $n$} is generally known as the \textsc{Maximum Induced Forest} problem and is known to be \Wone{}-hard with respect to the solution size~\cite{Khot2002}.
In the following, we consider \textsc{Feedback Vertex Set below Vertex Cover}, essentially the same problem but with a slightly smaller parameter.
We show that this change is sufficient to make the problem \NP-hard even if the parameter is fixed at two.

\newcommand{\instfvn}{\ensuremath{k}}

\begin{theorem}
	\label{thm:fvn-below-vc}
	\prob{Feedback Vertex Set below Vertex Cover} is $\NP$-hard even if~$\ell=\vc(G)-\fvn(G)=2$.
\end{theorem}
\begin{proof}
	We reduce from \prob{Feedback Vertex Set}.
	Let~$I = (G, \instfvn)$ be an instance of the latter problem.
	We assume without loss of generality that~$G$ has at least one edge.
	Let~$\lambda = \abs{V(G)} - \instfvn - 2$.
	We construct a graph~$H$ with~$\vc(H) = \abs{V(G)}$ that contains a feedback vertex set of size at most~$k' = \abs{V(G)}-2$ if and only if~$G$ contains a feedback vertex set of size at most~$\instfvn$.
	Note that then~$\ell = \vc(H) - k' = 2$.

	The construction of~$H$ is as follows:
	Add a copy of~$G$ to~$H$.
	Attach a leaf to every vertex~$v \in V(G)$, that is, add a new vertex~$u_v$ and the edge~$\{ v,  u_v \}$ to~$H$.
	Add a set~$V^*$ of~$\lambda$ vertices to~$V(H)$ and make each $u \in V^*$ adjacent to every vertex in~$V(G)$.

	Suppose $G$ contains a feedback vertex set~$S$ of size at most~$\instfvn$.
	Then the set~$S \cup V^*$ is a feedback vertex set of size at most~$\instfvn+\lambda = \abs{V(G)}-2 = k'$ for~$H$, as~$H-(S \cup V^*)$ is isomorphic to the forest~$G-S$ with a leaf attached to every vertex.

	Conversely, suppose that~$H$ contains a feedback vertex set~$S'$ of size at most~$k'$.
	We claim that there exists a feedback vertex set of size at most~$k'$ in~$H$ that contains none of the leaves~$u_v$ and all vertices in~$V^*$.
	Clearly, if $S'$ contains a leaf~$u_v$ attached to some~$v \in V(G)$, then~$(S' \setminus \{u_v\}) \cup \{v\}$ is also a feedback vertex set of size at most~$\abs{V(G)}-2$ in~$H$.
	Hence, we may assume that $V^* \subseteq S' \subseteq V^* \cup V(G)$.

	Next, suppose that $V^* \setminus S' \ne \emptyset$.
	If $|V^* \setminus S'| \ge 2$, then $S'$ contains at least $|V(G)| - 1$ vertices of $V(G)$ since otherwise two vertices in $V^* \setminus S'$ and two vertices in $V(G)$ form a cycle of length four.
	Note, however, that $|S'| \le k' = |V(G)| - 2$.
	Thus, we may assume that $|V^* \setminus S'| \le 1$.
	Towards showing that $V^* \setminus S' = \emptyset$, suppose that there is a vertex~$w \in V^* \setminus S'$.
	Then, for every edge~$\{ u, v \} \in E(G)$, we have~$u \in S'$ or $v \in S'$, as otherwise~$u$, $v$, and~$w$ induce a cycle in~$H-S'$.
	Thus, $S' \setminus V^*$ is a vertex cover of $G$.
	Let $x \in S' \cap V(G)$ be arbitrary.
	Such a vertex exists by the assumption that $G$ has at least one edge.
	We claim that $S^* = (S' \setminus \{x\}) \cup \{w\}$ is a feedback vertex set for~$H$ of size at most~$k'$.
	If it is not, then $H-S^*$ has a cycle that contains $x$.
	Let $y$ and $z$ be two neighbors of $x$ in this cycle.
	Note that $y, z \in V(G) \setminus S'$.
	Since $S' \cap V(G)$ is a vertex cover of $G$, it follows that the neighborhoods of $y$ and $z$ in $H - S^*$ are $\{ x, u_y \}$ and $\{ x, u_z \}$, respectively.
	The vertices $u_y$ and $u_z$ have degree one, and thus we have a contradiction to the existence of the aforementioned cycle.
	Thus, the claim follows.
	Lastly, having a solution~$S'$ with~$V^* \subseteq S' \subseteq V^*\cup V(G)$ of size at most~$k'$ implies that~$S = S' \setminus V^*$ is a feedback vertex set of size at most~$\instfvn$ for~$G$.

	Finally, to show that~$\ell = \fvn(H)-\vc(H)=2$, we need to show that $\vc(H) = \fvn(H) + 2 = \abs{V(G)}$.
	As~$V(G)$ is a vertex cover of~$H$, we have~$\vc(H) \le \abs{V(G)}$.
	Since $\{ \{ v, u_v \} \mid v \in V(G) \}$ is a matching of size $|V(G)|$ in $H$, we have $\vc(H) \ge \abs{V(G)}$.
\end{proof}

\section{Conclusion}
The goal of this work is to extend the above guarantee paradigm in parameterized complexity beyond the previously considered lower bounds on vertex cover, namely the maximum matching size and the optimal LP relaxation solution.
We approached this issue by considering various structural graph parameters that are upper-bounded by the vertex cover number.
This work sketches a rough contour of the parameterized complexity landscape of these kinds of parameterizations of both \textsc{Vertex Cover} and \textsc{Feedback Vertex Set}.
It raises a number of immediate open questions, of which we highlight four:
\begin{enumerate}[(i)]
	\item Is \textsc{Vertex Cover above Treewidth} also fixed-parameter tractable on arbitrary graphs?
	We also leave this question open for \textsc{Feedback Vertex Set}.
	\item In \cref{sec:vc-above-fv}, we showed that \textsc{Vertex Cover above Feedback Vertex Number} is \Wone-hard.
	A natural question to ask is whether this problem is NP-hard for a constant parameter value or whether it is in \XP{}, that is, whether it can be decided by an algorithm with running time~$\bigO(n^{f(\ell)})$ for an arbitrary computable function $f$.
	\item One can naturally generalize graph parameters like feedback vertex number or cluster deletion number by fixing a graph class $\calF$ and defining the $\calF$-free deletion number of any graph $G$ as the size of a smallest set $X\subseteq V(G)$ such that $G-X$ does not contain any $H\in \calF$ as an induced subgraph.
	If $\calF$ contains a graph that is not edgeless, then vertex cover number upper-bounds the $\calF$-free deletion number.
	It would be interesting to find a graph class $\calF$ such that \textsc{Vertex Cover above $\calF$-Free Deletion} is FPT or to rule out the existence of such a class.
	We have only answered this question (always in the negative) if $\calF$ is any of the following classes: all cycles, $\{P_3\}$, all complete graphs.
	\item Moving beyond parameterized complexity, can graphs in which the difference parameters we have considered are small be characterized in an elegant way?
	For instance, one can easily prove that $\vc(G) = \fvs(G)$ if and only if $G$ is edgeless.
	We are not aware of any simple characterization of graphs where $\vc(G) - \fvs(G) =1$ or~$\vc(G) - \fvs(G) \leq c$ for a larger constant $c$.
	Such a characterization could be useful for answering~(ii).
\end{enumerate}

{
\begingroup
  \let\clearpage\relax
  \renewcommand{\url}[1]{\href{#1}{$\ExternalLink$}}
  \bibliography{strings-long,bibliography}
\endgroup
}

\end{document}